\documentclass[twocolumn,longbibliography,prl,superscriptaddress,reprint]{revtex4-1}
\usepackage{amssymb}
\usepackage{graphicx}
\usepackage{dsfont}
\usepackage{amsmath}
\usepackage{bm}
\usepackage{times}
\usepackage{color}
\usepackage[applemac]{inputenc}
\usepackage[nice]{nicefrac}
\usepackage{amsthm}

\PassOptionsToPackage{normalem}{ulem}
\usepackage{ulem}

\graphicspath{{converted_graphics/}}

\theoremstyle{plain}
\newtheorem{thm}{\protect\theoremname}
  \theoremstyle{plain}
  \newtheorem{lem}{\protect\lemmaname}
  \theoremstyle{remark}
  \newtheorem{rem}{\protect\remarkname}

  \providecommand{\lemmaname}{Lemma}
  \providecommand{\remarkname}{Remark}
\providecommand{\theoremname}{Theorem}

\begin{document}

\title{Emergence of the pointer basis through the dynamics of correlations}

\begin{abstract}
We use the classical correlation between a quantum system being measured and its measurement apparatus to 
analyze the amount of information being retrieved in a quantum measurement process. Accounting for decoherence
of the apparatus, we show that these correlations may have a sudden transition from a decay regime
to a constant level. 
This transition characterizes a non-asymptotic emergence 
of the pointer basis, while the system-apparatus can still be quantum correlated. We provide a formalization of the concept of emergence of a pointer basis in an apparatus subject to decoherence. This contrast of the pointer basis emergence to the quantum to classical transition is demonstrated in an experiment with polarization entangled photon
pairs.

\end{abstract}

%%%%%%%%%%%%%%%%%%%%%%%%%%%%%%%%%
\author{M. F. Cornelio}
\affiliation{Instituto de  F\'\i sica, Universidade Federal de Mato Grosso, 78060–900, Cuiab\'a - MT, Brazil}
\affiliation{Instituto de  F\'\i sica Gleb Wataghin, Universidade Estadual de Campinas, 13083-859, Campinas, SP, Brazil}

%-------------------------------------------------------------------------------------
\author{O. Jim\'enez Far\'{\i}as}
\affiliation{Instituto de F\'{\i}sica, Universidade Federal do Rio de
Janeiro, Caixa Postal 68528, Rio de Janeiro, RJ 21941-972, Brazil}
\affiliation{Instituto de Ciencias Nucleares, Universidad Nacional Aut\'onoma de M\'exico
(UNAM), Apdo. Postal 70-543, M\'exico 04510 D.F.}
%--------------------------------------------------------------------------------
\author{F. F. Fanchini}
\affiliation{Departamento de F\'{\i}sica, Faculdade de Ci\^encias, Universidade Estadual Paulista, Bauru, SP, CEP 17033-360, Brazil}
%---------------------------------------
\author{I. Frerot}
\affiliation{D\'epartement de Physique, Ecole Normale Sup\'erieure, 24, rue Lhomond. 
F 75231 PARIS Cedex 05. }
%---------------------------------------
\author{G. H. Aguilar}
\affiliation{Instituto de F\'{\i}sica, Universidade Federal do Rio de
Janeiro, Caixa Postal 68528, Rio de Janeiro, RJ 21941-972, Brazil}
%---------------------------------------
\author{M. O. Hor-Meyll}
\affiliation{Instituto de F\'{\i}sica, Universidade Federal do Rio de
Janeiro, Caixa Postal 68528, Rio de Janeiro, RJ 21941-972, Brazil}
%---------------------------------------
\author{M. C. de Oliveira}
\affiliation{Instituto de  F\'\i sica Gleb Wataghin, Universidade Estadual de Campinas, 13083-859, Campinas, SP, Brazil}
\affiliation{Institute for Quantum Information Science, University of Calgary, Alberta T2N 1N4, Canada}

%----------------------------------------
\author{S. P. Walborn}
\affiliation{Instituto de F\'{\i}sica, Universidade Federal do Rio de
Janeiro, Caixa Postal 68528, Rio de Janeiro, RJ 21941-972, Brazil}
%---------------------------------------
\author{A. O. Caldeira}
\affiliation{Instituto de  F\'\i sica Gleb Wataghin, Universidade Estadual de Campinas, 13083-859, Campinas, SP, Brazil}
%---------------------------------------
\author{P. H. Souto Ribeiro}
\affiliation{Instituto de F\'{\i}sica, Universidade Federal do Rio de
Janeiro, Caixa Postal 68528, Rio de Janeiro, RJ 21941-972, Brazil}
%---------------------------------------

%\date{\today}
\maketitle

The measurement problem is at the core of fundamental questions of quantum physics and the quantum-classical boundary \cite{von,schr1,*schr2,*schr3}. One way to approach the classical limit 
is through the process of {\it decoherence} \cite{zurek1,*zurek2,*zurek3}, where
 a quantum measurement apparatus $\cal A$ interacts 
  with the system of interest $\cal S$.  The apparatus suffers decoherence through contact with the environment ($\cal E$) that collapses $\cal A$ into some classical set of {\it pointer states}, which are not altered by decoherence. 
The correlations between these states and the system are preserved, despite the dissipative decoherence process.
In this sense, decoherence selects the classical pointer states of $\cal A$, inducing a transition from quantum to classical states of the measurement apparatus.  The time scale associated with this transition is usually estimated by the decoherence half-life.    
In this work, we show that contrary to this idea, the pointer states can emerge in a well-defined
instant of time. This result is obtained by showing that the pointer basis emerges when the
classical correlation (CC) \cite{hender} between system and apparatus becomes constant.
 It emphasizes the importance of  CC
 in the investigation of the measurement process, 
even though the joint  $\cal SA$ state still has quantum features, as can be inferred by quantum discord \cite{oliv}. 
After the transition, measurements are repeatable being verifiable by other observers \cite{darwinism}, signalling the emergence of the pointer basis.
We demonstrate this behaviour experimentally using entangled photons \cite{almeida,ofarias}.

The discussion 
 starts by considering that a system 
$\cal S$ initially in a state $|\psi_s\rangle$ interacts with a measurement apparatus $\cal A$, so 
that they become entangled \cite{zurek1,*zurek2,zurek3,von,*schr1,*schr2,*schr3}.
The apparatus is in constant interaction with the environment $\cal E$, 
so that during the measurement process the composite system $\cal S+A+E$ evolves from the (uncoupled) initial state
$|\psi_s\rangle|A_0\rangle|E_0\rangle$ to %   \rightarrow
$\sum _i c_i |s_i\rangle |A_i\rangle |E_i(t)\rangle,$
{where $|A_i\rangle$ are orthogonal and thus distinguishable states of the apparatus, and $|E_i(t)\rangle$ are the states of the environment, which are inaccessible to the observer.}
The reduced density matrix of the system and the apparatus becomes
\begin{equation}
\rho_{sa} = \sum_{i,j} c_i c_j^{\ast} \langle E_j(t) | E_i(t)\rangle |s_i\rangle |A_i\rangle \langle s_j| \langle A_j|,
\label{eqab}
\end{equation}
where $\langle E_j(t) | E_i(t)\rangle$, with $i \neq j$, are rapidly decaying time-dependent coefficients.Therefore, after a characteristic period of time known as the decoherence time $\tau_D$, the resulting state of $\cal S + A$ is well approximated by
\begin{equation}
\rho_{sa} = \sum_{i} |c_i|^2 |s_i\rangle |A_i\rangle \langle s_i| \langle A_i|,
\label{eqabc}
\end{equation}
for which  the states of the bases $\{|s_i\rangle\}$ and  $\{|A_i\rangle\}$ are classically
correlated. 
This correlation permits an observer to obtain information about $\cal S$ via measurements on $\cal A$. 
In this sense, it is said that the environment selects a basis set of classical pointer states $\{|A_i\rangle\}$ of the apparatus 
and the decoherence time $\tau_D$ is traditionally recognized as a reasonable estimate of the time necessary for the pointer basis to emerge \cite{zurek3,brune96,*monroe96,*Paz93}. 
However, is it correct to assume that $\tau_D$ is the necessary time for the information about $\cal S$ be accessible to a classical observer?

To answer this question, let us consider the amount of information one obtains about the quantum system by observing the apparatus.
This information can be quantified by the CC defined in Refs. \cite{oliv,hender}.  {It tells us how much information one can retrieve about a first quantum system $\cal S$ through measurements performed on a second system $\cal A$.}
It is defined as the difference between the entropy of the system $\cal S$ and the average entropy conditioned to the output of measurements on  $\cal A$ \cite{oliv,hender}:
\begin{equation}
J_{s|\{\Gamma^a_i\}}(\rho_{sa}) = S(\rho_s) - \sum_i p_i S(\rho_s^i|\Gamma_i^a),
\label{eqc}
\end{equation}
where $S$ is the von Neumann entropy and $\{\Gamma^a_i\}$ is a complete positive operator valued measure (POVM) that determines which measurement is performed on $\cal A$.   In addition, as $J_{s|\{\Gamma^a_i\}}$ depends on the measurement chosen $\{\Gamma^a_i\}$, the maximum over all measurements on $\cal A$ determines the total  CC:
\begin{equation}
J_{s|a}^{\max}(\rho_{sa}) = \max_{\{\Gamma_{i}^{a}\}} \left[S(\rho_s) - \sum_i p_i S(\rho_s^i|\Gamma_i^a)\right].
\label{eqd}
\end{equation}
{ For solely classical-correlated states, this definition is equivalent to the mutual information $I(\rho_{sa}) = S(\rho_s) + S(\rho_a) - S(\rho_{sa})$, while for most quantum-correlated states it is not.
The difference between $I(\rho_{sa})$ and $J_{s|a}^{\max}$ is the so-called quantum discord (QD) \cite{oliv},
$\delta_{s|a}(\rho_{sa}) = I(\rho_{sa}) - J^{\max}_{s|a}(\rho_{sa})$, which quantifies genuine quantum correlations. This includes correlations that can be distinct from entanglement.}

Once $J^{\max}_{s|a}$ quantifies the information an observer retrieves about the quantum system by measuring her apparatus, it is natural to expect that its dynamics under decoherence can help us to understand the measurement process. 
In the following, we state two theorems that characterize the dynamics of  $J^{\max}_{s|a}$ and $J_{s|\{\Pi_i^a\}}$ during any decoherence process.
In the following, when we refer to decoherence, we mean a channel with a well defined and complete pointer basis. The complete proofs of the theorems are given in the Supplementary Information \cite{sup}.

{\em Theorem 1:\,\, Let $\rho_{sa}$ be the state of system-apparatus at a given moment. If the apparatus is subject to a decoherence process that leads to the projectors on the pointer basis $\{\Pi_i^a\}$, then  $J_{s|\{\Pi_i^a\}}$ is constant throughout the entire evolution.}

The interpretation of this theorem is clear. The classical correlations between the quantum system and the pointer states remain constant during the decoherence process. Therefore, if an observer monitors the apparatus through the pointer basis, she will always get the same information about the quantum system, independent of time and the decoherence rate.
In the next theorem, we will show that this is a particular property of the pointer basis.

{\em Theorem 2:\,\, Let $\rho_{sa}$ be the state of the system-apparatus at a given moment. 
If the apparatus is subject to a decoherence process leading to the projectors on the pointer basis $\{\Pi_{i}^{a}\}$ 
and $J_{s|\{\Pi_{i}^{a}\}}>0$, {then either (i) $J_{s|a}^{\max}$ is constant and equal to $J_{s|\{\Pi_{i}^{a}\}}$, or (ii) $J_{s|a}^{\max}$ decays monotonically to value $J_{s|\{\Pi_{i}^{a}\}}$ in a finite time, remaining constant and equal to $J_{s|\{\Pi_{i}^{a}\}}$ for the rest of the evolution.}}

Theorem 2 shows that $J^{\max}_{s|a}$ displays two kinds of evolution: {(i) constant or (ii) decay, which are fundamentally different in terms of the measurement process.}
During the decay process (ii), $J^{\max}_{s|a}$ is maximized in a basis of non-classical states, 
given by superpositions of the pointer states. 
Thus, the information about $\cal S$ that is available in the apparatus $J^{\max}_{s|a}$ is obtained by observation of non-classical {(non-pointer)} states, and it decays because these states are affected by decoherence. When $J^{\max}_{s|a}$ is constant (i), it is obtained by measurements in a basis of classical (pointer) states 
 not altered by decoherence.

In the decay case (ii) the transition from a decaying function to a constant cannot be analytical. {The switch is} signaled by a point of non-analiticity in the behavior of $J^{\max}_{s|a}$ as function of time.    
The study of non-analytic points of this sort, usually referred to as ``sudden changes", have already been reported in the literature \cite{xu}, 
and {typically focus on sudden changes in the quantum discord.}
However, as we show {here}, a {compelling and fundamental physical interpretation can be obtained by focusing rather on the sudden changes of the classical correlations in the context of the measurement process. } 

Now, consider an observer who was given the task to describe the quantum system $\cal S$,  
and suppose she can measure the apparatus $\cal A$ in any basis at any instant during the decoherence process.
Before the sudden change transition, the maximum information she  obtains, $J^{\max}_{s|a}$, is a decaying function of  time.
However, a classical description of a quantum system cannot depend on the decoherence rate of $\cal A$.
That is, for repeated tests of some fixed state of $\cal S$, the apparatus must always provide
the same information, otherwise the apparatus is useless.  {So}, from the perspective of the correlations between system and apparatus, {a meaningful pointer basis} cannot emerge before the transition. Indeed, after the transition, the maximum information $J^{\max}_{s|a}$ that can be obtained about $\cal S$ is given by measurement of $\cal A$ in the pointer basis. Thus, the time instant at which the classical correlations $J^{\max}_{s|a}$ become constant can be viewed as the emergence time of the pointer basis. 
Moreover, it is an immediate consequence of Theorem 2 that 
after the transition of $J^{\max}_{s|a}$ from decay to the constant regime,
the QD, i. e. quantum correlation, decays faster.
Nevertheless, due to the analyticity of $I(\rho_{sa})$, QD vanishes 
only in the asymptotic limit. Therefore, we prove the conjecture that QD
has no sudden death \cite{ferraro} for these channels.

To illustrate 
the consequence of Theorems 1 and 2, consider a  bipartite two-level (or two qubit) system $\cal S + A$ 
where the apparatus $\cal A$ is affected by the environment in two distinct ways: a phase damping (PD) and an amplitude 
damping channel (AD) \cite{salles}.
We consider two initial states (See Fig. \ref{dynamics}) for the joint system $\cal S + A$, both defined by the density matrix:
\begin{equation}
\rho_{sa}  =
\begin{pmatrix}
c&0&0&w\\0&b&z&0\\0&z&b&0\\w&0&0&c
\end{pmatrix}.\label{matrix}
\end{equation}
This state is an incoherent mixture of four Bell states.
For initial state 1, we choose $c=0.4$, $b =0.1$, $z =0.1$, and $w=0.4$, and for initial state 2 we take $c=0.4$, $b=0.1$, $z=0.1$, and $w=0.15$.
For these two cases, and considering the effects of PD and AD on the apparatus, we can consider only two measurements to maximize $J^{\max}_{s|a}$ during the dissipative dynamics: the projectors on the eigenstates of $\sigma_x$ and $\sigma_z$, $\{\Pi_x^a\}$ and $\{\Pi_z^a\}$ \cite{chen2011}.
 Fig. \ref{dynamics} displays the evolution of $J_{s|\{\Pi^a_x\}}$ and $J_{s|\{\Pi^a_z\}}$ as function of the time dependent parameter $p=(1-e^{-\gamma t})$, where $\gamma$ is the half-life of the decoherence.
The maximum information available $J^{\max}_{s|a}$ is given by the larger of these two quantities.
Fig. \ref{dynamics}(a) shows evolution under the action of the PD channel, for initial state 1. In this case the classical correlation $J_{s|\{\Pi^a_z\}}$ is constant, while the CC defined by measurements in any other basis decay.  The selection of the pointer basis emerges in the maximization of $J^{\max}_{s|a}$ through a sudden transition, occurring at a finite time ($p < 1$), at which point $J^{\max}_{s|a}=J_{s|\{\Pi^a_z\}}$, remains constant in the asymptotic limit. {The $\sigma_z$ eigenstates thus form the pointer basis for the PD channel.   In Fig. \ref{dynamics}(b), the trivial case for the initial condition 2 is shown. There $J^{\max}_{s|a}=J_{s|\{\Pi^a_z\}}$ is constant and maximum, and always corresponds to measurement in the pointer basis of $\sigma_z$ eigenstates. 

\begin{figure} % float placement: (h)ere, page (t)op, page (b)ottom, other (p)age
  \centering
   \includegraphics[width=9cm,height=4.55in,keepaspectratio]{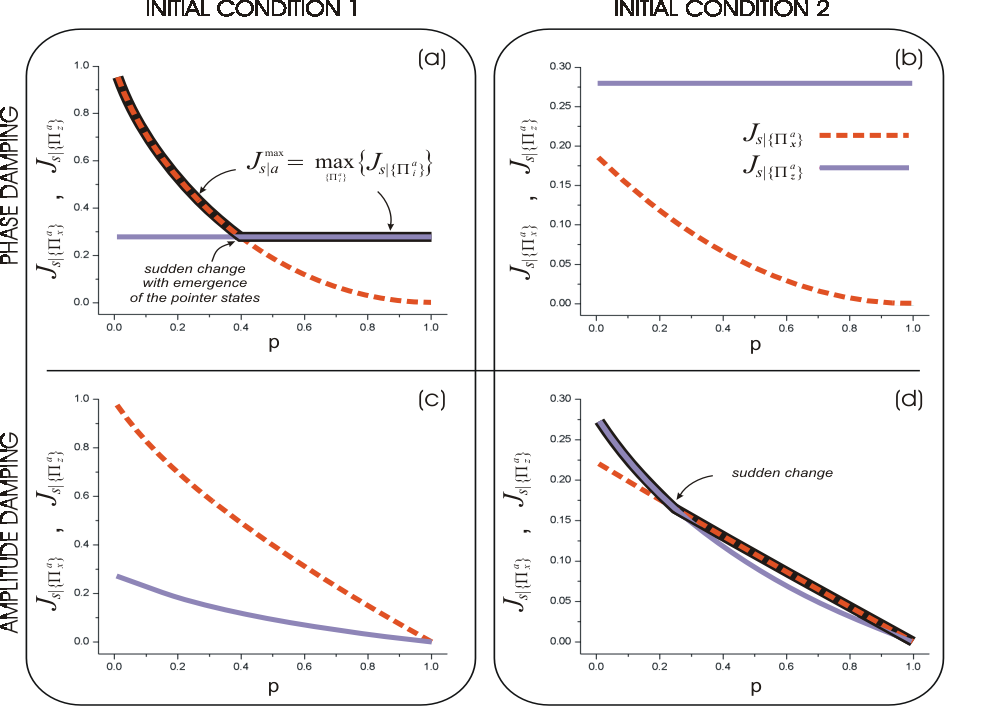}
  \caption{(Color online) Evolution of the quantity $J_{s|{\{\Pi_i^a\}}}$ calculated for
two different projectors: in the $\sigma_z$ basis (solid line) and in the $\sigma_x$ basis (dashed line).
}
\label{dynamics}
\end{figure}

Fig.\ref{dynamics}(c) shows the AD channel acting on initial state 1.  In this case,  $J^{\max}_{s|a}$ decays monotonically, since both correlations $J_{s|\{\Pi^a_x\}}$ and $J_{s|\{\Pi^a_z\}}$ decay and never cross. In Fig. \ref{dynamics}(d), we observe
that $J^{\max}_{s|a}$ suffers a sudden change, but continues to decay asymptotically. 
The behavior shown in Figs. \ref{dynamics}(c) and \ref{dynamics}(d) is due to the fact that the AD channel does not define a pointer basis.
Consequently, the CC for all possible bases of the apparatus $J_{s|\{\Pi^a_i\}}$ decay asymptotically to zero at some rate.
In other words, although there is a sudden transition from one basis to the other in the maximization of $J^{\max}_{s|a}$ in \ref{dynamics}(d), no preferred pointer basis of $\cal A$ is identified, since all correlations vanish asymptotically. The emergence of the pointer basis at a finite time can be attributed to the instant of time $\tau_E$ when a sudden transition occurs only when $J^{\max}_{s|a}$ remains constant after the transition.}
\par
For states like those in Eq. (\ref{matrix}), it is straightforward to find the emergence time $\tau_E$, namely when  $J_{s|\{\Pi^a_x\}}=J_{s|\{\Pi^a_z\}}$:
\begin{equation}
 \tau_E = \frac{1}{\gamma} \ln \left|\frac{z+w}{c-b} \right|.
\label{time} 
\end{equation}
Comparing this expression for $\tau_E$ with the decoherence half-life $\tau_D=1/\gamma$,
 shows that the pointer basis can emerge at times smaller or larger than $\tau_D$. 
This suggests that the decoherence half-life is not the best estimation for the 
{emergence of the pointer basis} of the apparatus $\cal A$. In particular, there are situations in which measurements in the pointer basis at time $\tau_D$  provide less information than other measurements.

%\emph{Experiment} - 
We performed an experiment to investigate the emergence of the pointer basis using polarization entangled photons produced in parametric down-conversion, subject to a phase damping channel \cite{salles,almeida,*ofarias}. The experimental setup is outlined in Fig. \ref{setup}.

\begin{figure} % float placement: (h)ere, page (t)op, page (b)ottom, other (p)age
  \centering
  \includegraphics[width=0.9\linewidth]{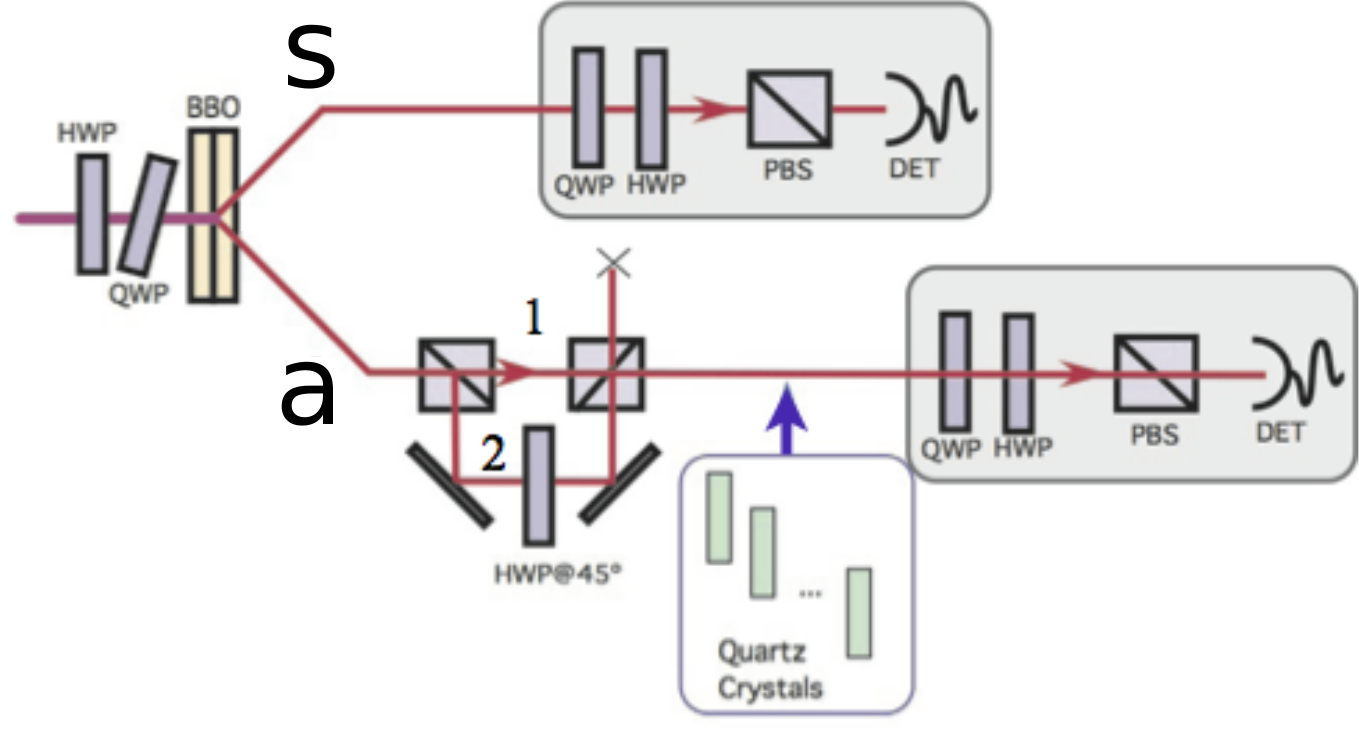}
  \caption{(Color online) Sketch of the experimental set-up. Polarization entangled photons
are generated  in a BBO crystal. States in the form of Eq. (\ref{matrix}) are produced using the
Mach-Zhender like scheme  in mode $a$, and dephasing channels are implemented
with birrefringent quartz crystals. Polarization state tomography is performed using the waveplates (QWP, HWP) and polarizing beam splitter (PBS).}
\label{setup}
\end{figure}

We produce polarization entangled photon pairs of the type
\begin{equation}
|\psi_1 \rangle = \alpha_1|H_s\rangle|H_a\rangle+\beta_1|V_s\rangle|V_a\rangle,
\label{exp1}
\end{equation}
where $|H\rangle$ and $|V\rangle$ are orthogonal polarization states, $\alpha_1$ and $\beta_1$ are complex coefficients, and the index $a$ and $s$ refer to apparatus and system respectively. The generation of the entanglement  between photons $s$ and $a$ represents the fast interaction between the system and the measurement apparatus.

To observe the sudden transition induced by decoherence, we produce an incoherent mixture of four Bell states, having the structure of Eq. (\ref{matrix}) \cite{xu}. To do so, the $s$ photon is sent directly to polarization analysis and detection, while the $a$ photon
is sent to an unbalanced Mach-Zehnder interferometer, as illustrated in Fig. \ref{setup}. Signal photons are split in modes $1$ and $2$ in a 50-50 beam splitter. Photons in mode 1 pass through the interferometer unchanged, while photons in mode $2$ propagate through a half wave plate which is oriented at 45$^{\circ}$. This operation switches polarization $|H\rangle$ to $|V\rangle$ and vice versa, producing the state
\begin{equation}
|\psi_2 \rangle = \alpha_2|H_s\rangle|V_{a2}\rangle+\beta_2|V_s\rangle|H_{a2}\rangle.
\label{exp2}
\end{equation}
The ratio, $\nicefrac{\alpha_i}{\beta_i}$, is controlled in the preparation of the initial state, while the ratio between the weights of the two Bell states is controlled with the neutral filter and the phase plate inside the interferometer.
 At the output, modes  1 and  2 are  recombined incoherently, due to the large relative path length difference.
 The state after the recombination is in the X-form of Eq. (\ref{matrix}).
 The coefficients can be controlled through the parameters of the initial state and those of the interferometer. We were able to produced states with fidelities as high as 0.98 compared to the target states.
The PD channel is implemented with successive 1mm long birefringent quartz crystals in mode $a$, which produces a relative delay between $H$ and $V$ polarization components.
When the temporal information is ignored, this corresponds to a dephasing channel.
Varying the total length of the quartz plates, we can control the
amount of decoherence induced, indexed by a parameter $p$,
ranging from $p = 0$ to $p = 1$ for which the coherence becomes zero.
The values of $p$ are obtained from process tomography, and the
classical correlations are obtained from the reconstructed density
matrix of the two-photon state.
From the reconstructed density matrix, we calculate the CC for different values of $p$.

\begin{figure} % float placement: (h)ere, page (t)op, page (b)ottom, other (p)age
  \centering
  \includegraphics[width=0.9\linewidth]{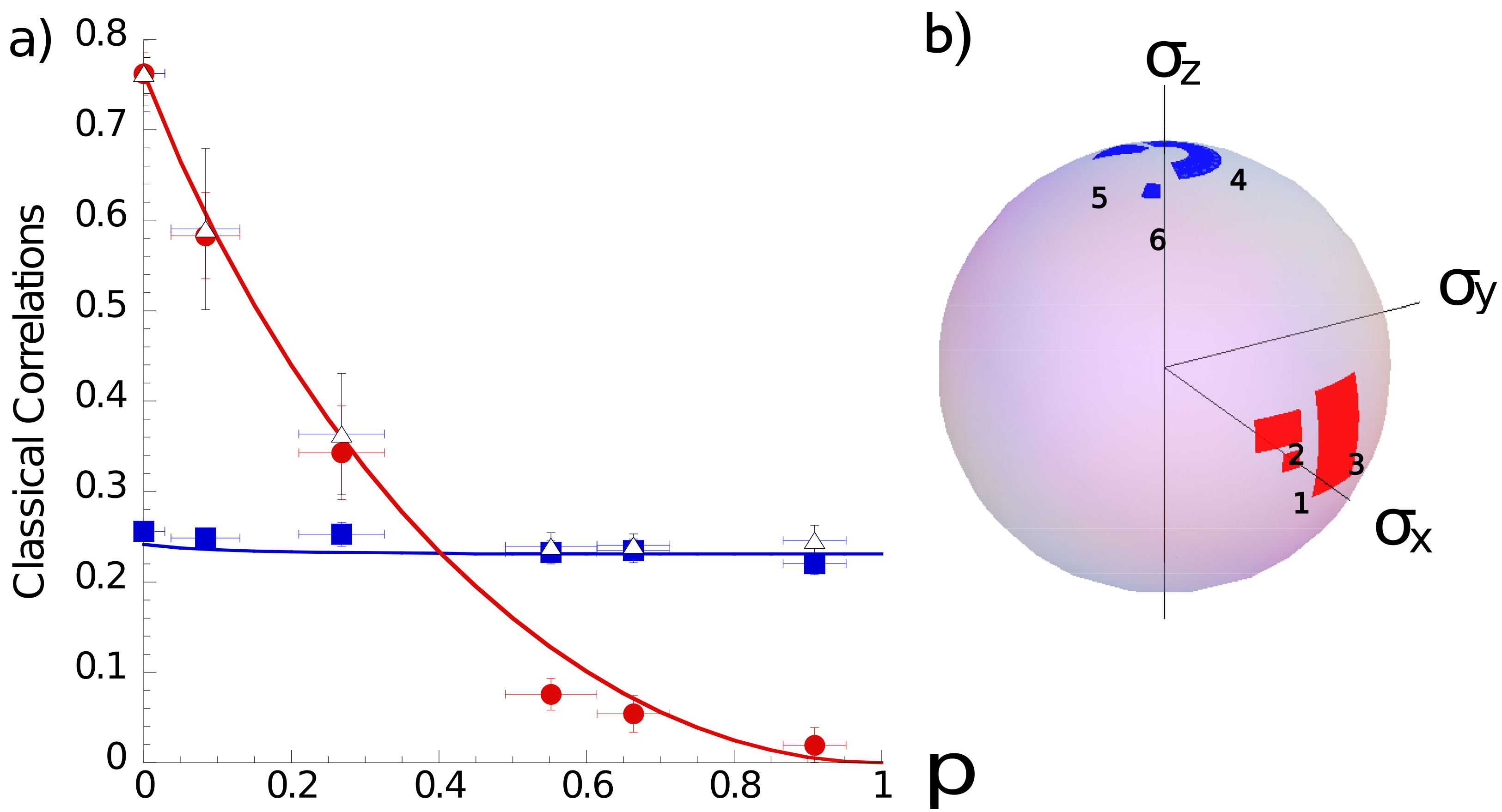} 
  \caption{(Color online) Experimental results for the phase damping channel showing emergence of the classical pointer states.  Blue squares and red circles correspond to the classical correlations $J_{s|\{\Pi^a_z\}}$ and $J_{s|\{\Pi^a_x\}}$, respectively.  Hollow triangles correspond to $J^{\max}_{s|a}$ obtained by maximizing over all possible projective measurements.  The inset shows the basis states used in obtaining $J^{\max}_{s|a}$ in the Bloch sphere, where it can be seen that the relevant basis jumps from $\sigma_x$ eigenstates to $\sigma_z$ eigenstates.  The solid lines correspond to theoretical calculations from the initial $p=0$ state.  Error bars were calculated using Monte Carlo simulations.}
  \label{results}
\end{figure}

Figure \ref{results} shows the measurement results for the CC,  $J_{s|\{\Pi^a_z\}}$  and $J_{s|\{\Pi^a_z\}}$ cross, after which $J^{\max}_{s|a}=J_{s|\{\Pi^a_z\}}=\mathrm{constant}$, signaling the appearance of the pointer basis for $p\geq 0.4$.  The maximum CC, obtained independently by maximizing $J^{\max}_{s|a}$ over all possible projective measurements, shows that either the $\sigma_x$ and $\sigma_z$ measurements optimize the available CC throughout the evolution.  The Fig. \ref{results}(b) shows the measurement directions in the Bloch sphere, where the jump from $\sigma_x$ to $\sigma_z$ occurs between the third and fourth data point.   The data shows that the pointer basis is selected well before the quantum discord vanishes, which occurs only in the asymptotic limit at $p=1$.

In conclusion, the measurement problem in quantum mechanics has been for long a difficult and debating subject. An enormous advance was obtained through the \emph{einselection} \cite{zurek1} program, shedding light on the obscure principles behind repeatability and predictability \cite{zurek1,darwinism} of measurements. Nonetheless, it has always been assumed that  \emph{the emergence of the pointer basis} for the apparatus and \emph{the emergence of the classicality} for the system+apparatus state were interchangeable phenomena. We have demonstrated in  Theorems 1 and 2 that this is not always true, with the formal definition for the emergence of the pointer basis in terms of the constancy of the classical correlation (CC) between system and apparatus. 
This constant profile for the correlations occurs whenever there exists projectors $\{\Pi_i^a\}$ commuting with the quantum map leading to the decoherence channel on the apparatus.
To understand why the constant behavior is required, we recall the fundamental principle that any measurement must be repeatable and verifiable by other observers -
according to the Copenhagen interpretation, reductions of the wave packet must return the same information any time and by any observers, a fact which is true only when the CC between system and apparatus is constant. In view of the interpretation given by  ``quantum Darwinism'' \cite{darwinism}, copies of the state of the apparatus will be reliable only when the CC is constant. In both ways the formal definition of the pointer basis, presented here, is suitable. 
We also observe that the time of the transition is usually different from the decoherence time.
Therefore, we must distinguish the classicality of the apparatus, which manifests itself through becoming a classical mixture of the pointer states due to decoherence, from the emergence of the pointer basis, which corresponds to another concept connected to repeatability and reproducibility of observations.
As we show in this work, these two phenomena are independent and usually occur at distinct stages of a measurement process.

\emph{Acknowledgements} - We acknowledge M. Piani and W. Zurek for fruitful discussions and comments on our manuscript.
Financial support for this work was provided by the Brazilian funding agencies CNPq,
CAPES, FAPESP and FAPERJ and the Instituto Nacional de Ci{\^e}ncia e Tecnologia - 
Informa\c{c}{\~a}o Qu{\^a}ntica (INCT-IQ). O.J.F was
funded by the Consejo Nacional de Ciencia y Tecnolog\'{\i}a, M\'exico. MCO acknowledges support by iCORE.

\section{Appendix: Proofs of the Theorems}

Here we show the proofs of the Theorems of the main paper. 
\begin{thm}
Let $\rho_{sa}$ be the state of system-apparatus at a given moment.
If the apparatus is subject to a decoherence process that leads to
a pointer basis $\{\Pi_{i}^{a}\}$, then $J_{s|\{\Pi_{i}^{a}\}}$
 is constant throughout the entire evolution.\end{thm}
\begin{proof}
By definition we have
\begin{equation}
J_{s|\{\Pi_{i}^{a}\}}=S(\rho_{s})-\sum_{i}p_{i}S(\rho_{i}^{s}|\Pi_{i}^{a}).\label{eq:Jpointer}
\end{equation}
 In addition, a dynamical map ${\cal E}$, representing a decoherence
process with pointer basis (PB) $\{\Pi^{i}\}$, may be written in the Kraus representation
as
\begin{equation}
{\cal E}_{q}(\rho_{sa})=(1-q)\rho_{sa}+q\sum_{i}\Pi_{i}^{a}\rho_{sa}\Pi_{i}^{a},\label{decMap}
\end{equation}
where $q$ is the probability that the measurement in the pointer
basis has happened. Usually, $q$ is an exponential function of time.
Now, the entropy $S(\rho_{s})$ in Eq. (\ref{eq:Jpointer}) cannot be
changed by a local operation on the apparatus. So, we have to prove
that the second term in Eq. (\ref{eq:Jpointer}) does not change with
$q$. Let us start by showing that each of the $S(\rho_{s}^{i}|\Pi_{i}^{a})$
does not change. By direct substitution of (\ref{decMap}) we have
\[
\Pi_{i}^{a}{\cal E}_{q}(\rho_{sa})\Pi_{i}^{a}=\Pi_{i}^{a}\rho_{sa}\Pi_{i}^{a}.
\]
That is, the state $\rho_{i}^{s}$ conditioned to the output $\Pi_{i}^{a}$
is independent of $q$ and, consequently, of time. So $S(\rho_{i}^{s}|\Pi_{i}^{a})$
also does not change. In addition we have
\[
p_{i}=\textrm{Tr}[\Pi_{i}^{a}{\cal E}_{q}(\rho_{sa})\Pi_{i}^{a}]=\textrm{Tr}[\Pi_{i}^{a}\rho_{sa}\Pi_{i}^{a}],
\]
which is also independent of $q$. So the probabilities $p_{i}$ in
Eq. (\ref{eq:Jpointer}) are constant. Therefore  $J_{s|\Pi_{i}^{a}}$
is constant.

\end{proof}
 Before proving Theorem 2, let us introduce and prove the following Lemma.
\begin{lem}
\emph{For a state $\rho_{sa}$ of the form
\begin{equation}
\rho_{sa}=\sum_{i}p_{i}\rho_{s}^{i}\otimes\Pi^{i},\label{eq:classico-quantico}
\end{equation}
 the maximization of $J_{s|a}$ is attained, and only attained, by the basis $\left\{ \Pi^{i}\right\}$.}\end{lem}
\begin{proof}
The lemma follows from the fact that the basis-dependent quantum discord
of a state $\sigma_{sa}$,	
\[
\delta_{s|\{\Gamma_{i}^{a}\}}(\sigma_{sa})=I(\sigma_{sa})-J_{s|\{\Gamma_{i}^{a}\}}(\sigma_{sa}),
\]
is zero if and only if the basis $\{\Gamma_{i}^{a}\}$ is such that
\cite{oliv} 
\[
\sigma_{sa}=\sum_{i}\Gamma_{i}^{a}\sigma_{sa}\Gamma_{i}^{a}.
\]

$(\Rightarrow)$ So using the basis $\left\{ \Pi^{i}\right\} $ for
$\rho_{sa}$, we get 
\[
\delta_{s|\{\Pi_{i}^{a}\}}=0
\]
 and $J_{s|\{\Pi_{i}^{a}\}}=I_{sa}$ which is the maximum possible
value for $J_{s|a}$.

$(\Leftarrow)$ Conversely, if we take a different basis $\{\Gamma_{i}^{a}\}$,
we have
\[
\rho_{sa}\neq\sum_{i}\Gamma_{i}^{a}\rho_{sa}\Gamma_{i}^{a}.
\]
So $\delta_{s|\{\Gamma_{i}^{a}\}}(\rho_{sa})>0$ and
\[
J_{s|\{\Gamma_{i}^{a}\}}(\rho_{sa})<I(\rho_{sa})=J_{s|\{\Pi_{i}^{a}\}}(\rho_{sa}).
\]
\end{proof}
\begin{thm}
Let $\rho_{sa}$ be the state of the system-apparatus at a given moment.
If the apparatus is subject to a decoherence process leading to the
pointer basis $\{\Pi_{i}^{a}\}$ and $J_{s|\{\Pi_{i}^{a}\}}>0$, then
either

(i) $J_{s|a}^{\max}$ is constant and equal to $J_{s|\{\Pi_{i}^{a}\}}$,
or

(ii) $J_{s|a}^{\max}$ decays monotonically to value $J_{s|\{\Pi_{i}^{a}\}}$
in a finite time, remaining constant and equal to $J_{s|\{\Pi_{i}^{a}\}}$
for the rest of the evolution.\end{thm}
\begin{proof}
Firstly, let us consider the case where initially $J_{s|a}^{\max}(\rho_{sa})=J_{s|\{\Pi_{i}^{a}\}}(\rho_{sa})$.
Since the decoherence map acts only on the apparatus, it is a
local operation and as such cannot increase the value $J_{s|a}$
\cite{hender}. On the other hand, $J_{s|\{\Pi_{i}^{a}\}}$ is constant
by Theorem 1 and, as $J_{s|a}$ cannot by smaller than $J_{s|\{\Pi_{i}^{a}\}}$
by definition, $J_{s|a}$ cannot decrease. Thus it is constant and
equal to $J_{s|\{\Pi_{i}^{a}\}}$ proving (i).

Now, let us consider the case where initially $J_{s|a}^{\max}(\rho_{sa})>J_{s|\{\Pi_{i}^{a}\}}(\rho_{sa})$.
Then there are some other bases $\{\Gamma_{i}^{a}\}$ such that 
\begin{equation}
J_{s|\{\Gamma_{i}^{a}\}}(\rho_{sa})>J_{s|\{\Pi_{i}^{a}\}}(\rho_{sa}).\label{eq:inequality}
\end{equation}
We will show that, for every basis $\{\Gamma_{i}^{a}\}$ satisfying
(\ref{eq:inequality}), $J_{s|\{\Gamma_{i}^{a}\}}$ will decay to
a smaller value them $J_{s|\{\Pi_{i}^{a}\}}(\rho_{sa})$ in a finite
time. To see this, we notice that, under a decoherence process described
by Eq. (\ref{decMap}), the state of the system-apparatus in the asymptotic
limit is given by
\[
\rho_{asym}=\sum_{i}\Pi_{i}^{a}\rho_{sa}\Pi_{i}^{a}=\sum_{i}\rho_{s}^{i}\otimes\Pi_{i}^{a},
\]
where $\rho_{s}^{i}$ is the state of the system conditioned to the
output $\Pi_{i}^{a}$. Now we can see that $\rho_{asym}$ is of the
form (\ref{eq:classico-quantico}) and, by Lemma 1, $J_{s|\{\Gamma_{i}^{a}\}}(\rho_{asym})$
is strictly smaller than $J_{s|\{\Pi_{i}^{a}\}}(\rho_{asym})$. So
we have
\begin{equation}
J_{s|\{\Pi_{i}^{a}\}}(\rho_{sa})>J_{s|\{\Gamma_{i}^{a}\}}(\rho_{asym}).\label{eq:inequalityAsymp}
\end{equation}
Therefore, the dynamics of correlations in a basis $\{\Gamma_{i}^{a}\}$,
given by $J_{s|\{\Gamma_{i}^{a}\}}({\cal E}_{q}(\rho_{sa}))$, satisfies
inequality (\ref{eq:inequality}) for $q=0$ and inequality (\ref{eq:inequalityAsymp})
for $q=1$. As $J_{s|\{\Gamma_{i}^{a}\}}({\cal E}_{q}(\rho_{sa}))$
is an analytical function of $q$, it must satisfies (\ref{eq:inequalityAsymp})
also in a neighborhood of $q=1$ and $J_{s|\{\Gamma_{i}^{a}\}}({\cal E}_{q}(\rho_{sa}))$
can be equal to $J_{s|\{\Pi_{i}^{a}\}}(\rho_{sa})$ only for some
$q$ strictly smaller than 1 which represents a finite time.

The dynamics of $J_{s|a}$ is a maximization over all the bases $\{\Gamma_{i}^{a}\}$ and $\{\Pi_{i}^{a}\}$. So $J_{s|a}$ will decay
to the value $J_{s|\{\Pi_{i}^{a}\}}(\rho_{sa})$ in a finite time.
To see that this decay is monotonic, we notice that the decoherence
map in Eq. (\ref{decMap}) is a local operation on the apparatus alone and
thus cannot increase the classical correlation. 
\end{proof}

\begin{rem}
One might think that the condition $J_{s|\{\Pi_{i}^{a}\}}>0$ may not
be necessary, but it is. In fact, we can write down a state with $J_{s|\{\Pi_{i}^{a}\}}=0$
that, under a dephasing channel, $J_{s|a}$ decays to zero asymptoticly
without getting constant. For example, consider the state
\[
\rho=\frac{I}{4}+\frac{1}{4}\sigma_{x}\otimes\sigma_{x},
\]
where $I$ is the identity operator $\sigma_{x}$ is the usual Pauli
matrix. In this case, $J_{s|\{\Pi_{i}^{a}\}}=0$ where the PB is made
of the eigenvectors of the $\sigma_{z}$ operator, and $J_{s|a}$ decays
asymptotically to zero. Of course, this case has no meaning for the
measurement process since no correlation is preserved between the
system and the apparatus. \end{rem}

\end{document}